\documentclass{LMCS}

\def\dOi{10(1:8)2014}
\lmcsheading%
{\dOi}
{1--20}
{}
{}
{Jan.~13, 2013}
{Feb.~12, 2014}
{}


\ACMCCS{[{\bf Theory of computation}]: Formal languages and automata
  theory---Automata extensions}

\usepackage {amsfonts}
\usepackage {amssymb,amsmath,amsthm}
\usepackage {tikz}

\usepackage {xspace}
\usepackage {url}

\usepackage{hyperref}

\usepackage {tikz}

\usetikzlibrary{positioning}
\usetikzlibrary{decorations.pathreplacing}
\usetikzlibrary{automata,positioning}
\usetikzlibrary{decorations.pathmorphing}
\usetikzlibrary{decorations.markings}
\usetikzlibrary{arrows}

\newcommand{\operator}[2]{\newcommand{#1}{{{\ensuremath{\mathop{\mathrm{#2}}}\xspace}}}}

\newcommand{\domain}[2]{\newcommand{#1}{\ensuremath{\mathbb {#2}}}\xspace}

\operator{\fo}{FO}
\operator{\mso}{MSO}
\operator{\wmso}{WMSO}
\operator{\msou}{MSO+U}
\operator{\wmsou}{WMSO+U}
\operator{\wmsor}{WMSO+R}

\newcommand{\wB}{\ensuremath{\omega \mathrm{B}}\xspace}
\newcommand{\wS}{\ensuremath{\omega \mathrm{S}}\xspace}

\newcommand{\wT}{\ensuremath{\omega \mathrm{T}}\xspace}

\newcommand{\fT}{\ensuremath{\mathrm{T}}\xspace}
\newcommand{\fB}{\ensuremath{\mathrm{B}}\xspace}
\newcommand{\fS}{\ensuremath{\mathrm{S}}\xspace}

\domain{\N}{N}
\domain{\Z}{Z}
\domain{\Q}{Q}
\domain{\R}{R}
\domain{\C}{C}

\newcommand{\closure}{\mathrm{Closure}}

\newcommand{\prom}[1]{\ensuremath{\widehat{{#1}^\ast}}\xspace}

\newcommand{\aut}[1]{\ensuremath{\mathcal {#1}}\xspace}

\newcommand{\mathbold}[3]{\ensuremath{\mathbf{#1}^{#2}_{#3}}}

\newcommand{\bsigma}[1]{\mathbold{\Sigma}{0}{#1}}
\newcommand{\bpi}[1]{\mathbold{\Pi}{0}{#1}}

\newcommand{\asigma}[1]{\mathbold{\Sigma}{1}{#1}}
\newcommand{\api}[1]{\mathbold{\Pi}{1}{#1}}

\newcommand{\comment}[1]{}

\newcommand{\fun}[3]{\ensuremath{#1\colon #2 \to #3}}

\operator{\dom}{dom}
\operator{\id}{id}
\operator{\sgn}{sgn}

\operator{\close}{cl}
\operator{\inter}{int}

\newcommand{\comp}[1]{\ensuremath{\overline{#1}}}

\newcommand{\newthm}[2]{\newtheorem{#1}[test_thereom]{#2}}

\newcommand{\ident}[2]{\newcommand{#1}{\ensuremath{\mathrm{#2}}\xspace}}

\ident{\lang}{L}

  \newthm{theorem}{Theorem}
  \newthm{lemma}{Lemma}
  \newthm{definition}{Definition}
  \newthm{problem}{Problem}
  \newthm{proposition}{Proposition}
  \newthm{example}{Example}
  \newthm{corollary}{Corollary}
  \newthm{claim}{Claim}

  \newthm{observation}{Observation}
  \newthm{remark}{Remark}
  \newthm{assignment}{Assignment}
  \newthm{conjecture}{Conjecture}
  \newthm{note}{Note}

\title
[Separation property for \wB- and \wS-regular languages]
{Separation property for \wB- and \wS-regular languages}

\newcommand{\thankncn}{\thanks{Work supported by the National Science Center (decision 
DEC-2012/07/D/ST6/02443)}}

\author{Micha{\l} Skrzypczak}
\thankncn
\address{Institute of Informatics, University of Warsaw, ul. Banacha 2, 02-097 Warsaw, Poland}
\email{mskrzypczak@mimuw.edu.pl}

\newcommand{\trans}[1]{\overset{#1}{\longrightarrow}}

\newcommand{\ceps}{\ensuremath{\mathbf{nil}}}
\newcommand{\cinc}{\ensuremath{\mathbf{inc}}}
\newcommand{\cres}{\ensuremath{\mathbf{reset}}}

\newcommand{\cop}{o}

\renewcommand{\setminus}{-}

\operator{\traces}{traces}

\operator{\sep}{sep}

\newcommand{\types}[1]{\ensuremath{\mathrm{Tp}^{#1}}\xspace}

\newcommand{\mtraces}{M_{\mathrm{trans}}}

\newcommand{\val}[1]{\mathrm{val}({#1})}

\newcommand{\shl}{\leftarrow}
\newcommand{\shr}{\rightarrow}

\newcommand{\finA}{w}
\newcommand{\finB}{z}
\newcommand{\seqA}{W}
\newcommand{\seqB}{Z}
\newcommand{\seqC}{Y}
\newcommand{\infA}{u}
\newcommand{\infB}{v}

\newcommand{\etype}{\mathrm{Etp}}

\begin{document}

% ABSTRACT =====================================================================

\begin{abstract}
\noindent In this paper we show that \wB- and \wS-regular languages satisfy the following separation-type theorem
\begin{quote}
If $L_1,L_2$ are disjoint languages of $\omega$-words both recognised
by \wB- (resp. \wS)-automata then there exists an $\omega$-regular
language $L_\sep$ that contains $L_1$, and whose complement contains $L_2$.
\end{quote}

In particular, if a language and its complement are recognised by \wB- (resp. \wS)-automata then the language is $\omega$-regular.

The result is especially interesting because, as shown by Boja{\'n}czyk and Colcombet, \wB-regular languages are complements of \wS-regular languages. Therefore, the above theorem shows that these are two mutually dual classes that both have the separation property. Usually (e.g. in descriptive set theory or recursion theory) exactly one class from a pair $\mathcal C, \mathcal C^c$ has the separation property.

The proof technique reduces the separation property for $\omega$-word languages to profinite languages using Ramsey's theorem and topological methods. After that reduction, the analysis of the separation property in the profinite monoid is relatively simple. The whole construction is technically not complicated, moreover it seems to be quite extensible.

The paper uses a framework for the analysis of \fB- and \fS-regular languages in the context of the profinite monoid that was proposed by Toru{\'n}czyk.
\end{abstract}

\maketitle

% INTRODUCTION =================================================================

\section{Introduction}

The classes of \wB- and \wS-regular languages are extensions of $\omega$-regular languages proposed by Boja{\'n}czyk and Colcombet in~\cite{bojanczyk_bounds}. The idea is to define \emph{asymptotic} properties of $\omega$-words. The standard example is the following \wB-regular language
\[\left\{a^{n_0}b a^{n_1}ba^{n_2}b\ldots:\ \text{the sequence $n_i$ is bounded}\right\}\subseteq\{a,b\}^\omega.\]

The main technical contribution of~\cite{bojanczyk_bounds} is the following theorem.
\begin{theorem}[Theorem~4.1 in~\cite{bojanczyk_bounds}]\label{th:duality}
The complement of an \wB-regular language is effectively \wS-regular and vice versa.
\end{theorem}

In this paper we show that both these classes admit the separation property. In general, a class of languages $\mathcal C$ has the \emph{separation property} with respect to a class $\mathcal D$, if the following condition holds:

\begin{quote}
For every pair of disjoint languages $L_1,L_2$ from $\mathcal C$ there exists a language $L_\sep\in\mathcal D$ such that\footnote{To distinguish the complement from the closure we denote the complement of a set $X$ by $X^c$.}
\[L_1\subseteq L_\sep\quad\text{and}\quad L_2\subseteq L_\sep^c.\]
\end{quote}

In that case we say that $L_\sep$ \emph{separates} $L_1$ and $L_2$. If not mentioned otherwise, the class $\mathcal D$ is taken as $\mathcal C\cap \mathcal C^c$.

Usually, one class from a pair of dual classes $\mathcal C, \mathcal C^c$ has the separation property and the other one does not. Below we recall some known separation-type theorems. The first one is a simple observation about Borel sets.

\begin{theorem}[Theorem~II~22.16 in~\cite{kechris_descriptive}]
Let $\eta<\omega_1$. Every two disjoint $\bpi \eta$ languages can be separated by a language that belongs to $\bpi \eta\cap \bsigma \eta$. On the other hand, there exists a pair of disjoint languages in $\bsigma \eta$ that cannot be separated as above.
\end{theorem}

The following theorem is an important extension to the projective hierarchy.

\begin{theorem}[Lusin (see~\cite{kechris_descriptive}]
If $L_1,L_2\in\asigma 1$ are two disjoint analytic sets then there exists a Borel set separating them. There exists a pair of disjoint co-analytic (i.e. $\api 1$) sets that cannot be separated by any Borel set.
\end{theorem}

The above theorem has its counterpart for regular languages of infinite trees. There is a correspondence between the parity index of a tree automaton and the topological complexity of the tree-language recognised by it. In particular, alternating $(1,2)$-parity tree automata (ATA) recognise analytic languages, $(0,1)$-parity ATA recognise co-analytic sets, while weak alternating automata recognise only Borel languages. The following two theorems are analogous to the above one in the tree-regular context.

\begin{theorem}[Rabin~\cite{rabin_separation}, also Kupferman, Vardi~\cite{kupferman_complementation}]\label{th:rabin_sep}
If $L_1,L_2$ are two disjoint regular tree languages recognised by $(1,2)$-parity ATA then there exists a language separating them recognisable by a weak alternating automaton (thus Borel).
\end{theorem}

\begin{theorem}[Hummel, Michalewski, Niwi{\'n}ski~\cite{hummel_separation}]
There exists a pair of tree-languages recognised by $(0,1)$-parity ATA that cannot be separated by any Borel set.
\end{theorem}

The above theorem is extended for higher levels of the alternating index hierarchy in~\cite{michalewski_separation}. Theorem~\ref{th:rabin_sep} relies on the fact that every alternating $(1,2)$-parity tree automaton is equivalent to a nondeterministic $(1,2)$-parity tree automaton.

In this work we show that both classes of \wB- and \wS-regular languages have the separation property with respect to $\omega$-regular languages. The proposed constructions are effective. The result is especially interesting since these are two mutually dual classes (see Theorem~\ref{th:duality} above). As a consequence of the separation properties we obtain the following corollary.

\begin{corollary}
If a given language of $\omega$-words $L$ and its complement $L^c$ are both \wB-regular (resp. \wS-regular) then $L$ is $\omega$-regular.
\end{corollary}

The above result in the case of \wB-regular languages was independently known by some researchers in the area. Nevertheless, to the best of the author's knowledge, it has never been published.

To prove the main result we reduce the separation property of $\omega$-word languages to the case of profinite words. For this purpose we use \fB- and \fS-automata introduced in~\cite{colcombet_stabilisation}. As shown in~\cite{torunczyk_limitedness} it is possible to define a language recognised by a \fB- or \fS-automaton as a subset of the profinite monoid $\prom{A}$. An intermediate step in our reasoning is proving the separation property for \fB- and \fS-regular languages of profinite words.

The paper is organised as follows. In Section~\ref{s:basic} we introduce basic notions. Section~\ref{s:automata} defines the automata models we use. In Section~\ref{s:profinite_sep} we prove separation results for languages of profinite words recognised by \fB- and \fS-automata. Section~\ref{s:reduction} contains the crucial technical tool, Theorem~\ref{th:recognition}, that enables to transfer separation results for languages of profinite words to the case of $\omega$-words. In Section~\ref{s:omega_sep} we use this theorem to show that \wB- and \wS-regular languages have the separation property. Section~\ref{s:direct} contains a direct and simpler proof of the separation property for the \wB-regular case. This proof was proposed by Thomas Colcombet, we present it here with his kind permission. Finally, in Section~\ref{s:ack} we give acknowledgements.

% BASIC NOTIONS ================================================================

\section{Basic notions}\label{s:basic}

We work with two models (\wB and \wS) at the same time. Therefore, we introduce a notion \wT to denote one of the models: \wB or \wS. By \fT we denote the corresponding finite word automata (\fB or \fS). By $A$ we denote a finite alphabet. Elements of $A^\ast$ are called finite words while $A^\omega$ is the set of $\omega$-words.

%By $\finA,\finB$ we denote finite words over this alphabet, letters $\infA,\infB$ denote $\omega$-words over $A$, by $\seqA,\seqB$ we denote sequences of finite words: $\seqA=\finA_0,\finA_1,\ldots$

\subsection{Monoids}

We use monoids and Ramsey's theorem to decompose $\omega$-words into finite ones.

\begin{definition}
A (finite) \emph{monoid} is a (finite) algebraic structure $M$ equipped with an operation $\fun{\cdot}{M^2}{M}$ that is associative ($a\cdot(b\cdot c)=(a\cdot b)\cdot c$) and with a distinguished element $1\in M$ that satisfies $1\cdot a = a\cdot 1 = a$.

The operation $\cdot$ is called \emph{product} and $1$ is called the \emph{neutral element}.

An element $e\in M$ is called \emph{idempotent} if $e\cdot e = e$.
\end{definition}

Observe that the set of all finite words $A^\ast$ has a natural structure of monoid with the operation of concatenation and $1$ defined as the empty word.

\begin{definition}
A function $\fun{h}{M}{N}$ is a \emph{homomorphism} between monoids $M,N$ if it preserves the product:
\[h(s\cdot t) = h(s)\cdot h(t).\]
\end{definition}

Now we define the monoid representing possible \emph{runs} of a nondeterministic automaton. It can be seen as an algebraic formalisation of the structure used by B\"uchi~\cite{buchi_decision} in his famous complementation lemma.

\begin{definition}
Let $\aut{A}$ be a nondeterministic automaton with states $Q$. Define $\mtraces(\aut{A})$ as $\mathcal P(Q\times Q)$. Let the neutral element be $\{(q,q):q\in Q\}$ and product:
\[s\cdot t\ =\ \left\{(p,r):\ \exists_{q\in Q}\ (p,q)\in s\ \wedge\ (q,r)\in t\right\}.\]

Let $\fun{h_{\aut{A}}}{A^\ast}{\mtraces(\aut{A})}$ map a given finite word $\finA$ to the set of pairs $(p,q)$ such that the automaton $\aut{A}$ has a run over $\finA$ starting in $p$ and ending in $q$.
\end{definition}

It is easy to check that $\mtraces(\aut{A})$ is a finite monoid and $h_\aut{A}$ is a homomorphism.

\subsection{Profinite monoid}

In this subsection we introduce the profinite monoid $\prom{A}$. A formal introduction to profinite structures can be found in~\cite{almeida_profinite} or~\cite{pin_profinite}. We refer to~\cite{pin_profinite}.

First we provide a construction of the profinite monoid $\prom{A}$. The idea is to enhance the set of all finite words by some \emph{virtual} elements representing sequences of finite words that are more and more similar.

Let $K_0,K_1,\ldots$ be a list of all regular languages of finite words. Let $X=2^\omega$. Each element $x\in X$ can be seen as a sequence of bits, the bit $x(n)$ indicates whether our \emph{virtual} word belongs to the language $K_n$.

Define $\fun{\mu}{A^\ast}{X}$ by the following equation:
\[\mu(\finA)_n\ =
\begin{cases} 1 & \text{if $w\in K_n$,}\\
0 &\text{if $w\notin K_n$.}
\end{cases} \]

The function $\mu$ defined above is an embedding of $A^\ast$ into $X$. Let $\prom{A}\subseteq X$ be the closure of $\mu(A^\ast)$ in $X$ with respect to the product topology of $X$. Therefore, $\prom{A}$ contains $\mu(A^\ast)$ and the limits of its elements. To simplify the notion we identify $w\in A^\ast$ with its image $\mu(w)\in\prom{A}$.

\begin{example}[Proposition~2.5 in~\cite{pin_profinite}]\label{ex:factorial}
Let $w_n=a^{n!}$ for $n\in\N$. A simple automata-theoretic argument shows that for every regular language $K$, either almost all words $(w_n)_{n\in \N}$ belong to $K$ or almost all do not belong to $K$. Therefore, the sequence $\left(\mu(w_n)\right)_{n\in\N}$ is convergent coordinate-wise in $X$. The limit of this sequence is an element of $\prom{A}\setminus \mu\left(A^\ast\right)$.
\end{example}

The following fact summarises basic properties of $\prom{A}$.

\begin{fact}[Proposition~2.1, Proposition~2.4, and Theorem~2.7 in~\cite{pin_profinite}]
$\prom{A}$ is a compact metric space. $A^\ast$ (formally $\mu\left(A^\ast\right)$) is a countable dense subset of $\prom{A}$. %The set of \emph{virtual} profinite words $\prom{A}\setminus A^\ast$ is homeomorphic to the Cantor set.
$\prom{A}$ has a structure of a monoid that extends the structure of $A^\ast$ and the concatenation is continuous.
\end{fact}

It turns out that the operation assigning to every regular language of finite words $K\subseteq A^\ast$ its topological closure $\comp{K}\subseteq \prom{A}$ has good properties (see Theorem~\ref{th:iso}). Therefore, we introduce the following definition.

\begin{definition}
A \emph{profinite-regular} language is a subset of $\prom{A}$ of the form $\comp{K}$ for some regular language $K\subseteq A^\ast$.
\end{definition}

Using this definition, we can denote a generic profinite-regular language as $\comp{K}$ for $K$ ranging over regular languages. Using the definition of $\mu$ one can show the following easy fact.

\begin{fact}\label{ft:coordinates}
A language of profinite words $M\subseteq\prom{A}$ is profinite-regular if and only if it is of the form
\begin{equation}
M=\left\{x\in X:\ x\in\prom{A}\ \wedge\ x_n = 1\right\},
\label{eq:coords}
\end{equation}
for some $n\in\N$. In that case $M=\comp{K_n}$.
\end{fact}

The structures of profinite-regular and regular languages are in some sense identical. This is expressed by the following fact.

\begin{theorem}[Theorem~2.4 in~\cite{pin_profinite}]\label{th:iso}
The function $K\mapsto \comp{K}\subseteq\prom{A}$ is an isomorphism of the Boolean algebra of regular languages and the Boolean algebra of profinite-regular languages. Its inverse is $M\mapsto \mu^{-1}(M)\subseteq{A^\ast}$ (when identifying $A^\ast$ with $\mu(A^\ast)$ we can write $M\mapsto M\cap A^\ast\subseteq A^\ast$).
\end{theorem}

By the definition of $\prom{A}$ and fact that regular languages are closed under finite intersection, we obtain the following important fact.

\begin{fact}\label{ft:basis}
The family of regular languages of profinite words is a basis of the topology of $\prom{A}$.
\end{fact}

The topology of $\prom{A}$ is the product topology. Therefore, a sequence of finite words $\seqA=\finA_0, \finA_1,\ldots$ is convergent to $\finA\in\prom{A}$ if and only if $\left(\mu(\finA_n)\right)_{n\in\N}\subseteq X$ is convergent coordinate-wise to $\finA$. The following fact formulates this condition in a more intuitive way.

\begin{fact}\label{ft:conv_seq}
A sequence of finite words $\seqA=\finA_0,\finA_1,\ldots$ is convergent to $\finA\in\prom{A}$ if and only if for every profinite-regular language $\comp{K}$ either:
\begin{itemize}
\item $\finA\in \comp{K}$ and almost all words $\finA_n$ belong to $K$,
\item $\finA\notin \comp{K}$ and almost all words $\finA_n$ do not belong to $K$.
\end{itemize}
\end{fact}

The topology of $\prom{A}$ is defined in such a way that it corresponds precisely to profinite-regular languages. The following fact summarises this correspondence.

\begin{fact}[Proposition~4.2 in~\cite{pin_profinite}]\label{ft:clopen}
A language $M\subseteq\prom{A}$ is profinite-regular if and only if it is a closed and open (clopen) subset of $\prom{A}$.
\end{fact}

\begin{proof}
First assume that $M=\comp{K}$ is a regular language of profinite words. Equation \eqref{eq:coords} in Fact~\ref{ft:coordinates} defines a closed and open set.

Now assume that $M$ is a closed and open subset of $\prom{A}$. Recall that profinite-regular languages form a basis for the topology of $\prom{A}$ (Fact~\ref{ft:basis}). Since $M$ is open so it is a union of base sets $\bigcup_{j\in J}\ \comp{K_j}$. Since $M$ is a closed subset of a compact space $\prom{A}$, $M$ is compact. Therefore, only finitely many languages among $\left\{\comp{K_j}\right\}_{j\in J}$ form a cover of $M$. But a finite union of profinite-regular languages is a profinite-regular language. Therefore, $M$ is profinite-regular.
\end{proof}

\subsection{Ramsey-type arguments}

In this subsection we recall Ramsey's theorem and show its application to finite monoids. This technique was used by B\"uchi in his complementation lemma~\cite{buchi_decision}. Additionally, we recall some extensions of Ramsey's theorem to compact spaces. In the following, by $[\N]^2$ we denote the set of all unordered pairs of natural numbers.

\begin{theorem}[Ramsey]\label{th:ramsey}
Assume that $\fun{\alpha}{[\N]^2}{C}$ is a function that assigns to every pair of numbers $\{n,m\}\in[\N]^2$ a \emph{colour} $\alpha(\{n,m\})\in C$. Additionally, assume that the set of colours $C$ is finite. Then there exists an infinite \emph{monochromatic} set $S\subseteq\N$: a set $S$ such that for some colour $c\in C$ and every pair of numbers $\{n,m\}\subset S$ we have
\[\alpha(\{n,m\})=c.\]
\end{theorem}

The following theorem shows an application of Ramsey's theorem to the $\omega$-word case.

\begin{theorem}\label{th:ramsey_decomposition}
Let $M$ be a finite monoid and $\fun{h}{A^\ast}{M}$ be a homomorphism. Then for every $\omega$-word $\infA\in A^\omega$ there exists a sequence of finite words $\finA_0,\finA_1,\finA_2,\ldots$ and two elements $s,e$ of the monoid $M$ such that:
\begin{enumerate}[label=(\roman*)]
\item $\infA=\finA_0 \finA_1 \finA_2\ldots$,
\item $h(\finA_0) = s$,
\item $h(\finA_n) = e$ for every $n > 0$,
\item $s\cdot e = s$ and $ e\cdot e = e$.
\end{enumerate}
\end{theorem}

A pair $(s,e)$ satisfying the above constraints is often called a \emph{linked pair}, see~\cite{perrin_pin_words}. To simplify the properties in the above theorem we introduce the following definition.

\begin{definition}
For a given homomorphism $\fun{h}{A^\ast}{M}$ we say that the \emph{type} (or \emph{$h$-type}) of a decomposition $\infA=\finA_0\finA_1\ldots$ is $t= (s,e)$ if: $s\cdot e = s$, $e\cdot e=e$, $h(\finA_0)=s$, and $h(\finA_n)=e$ for all $n>0$.
\end{definition}

Using the above definition we can restate Theorem~\ref{th:ramsey_decomposition} as: for every $\omega$-word $\infA$ and homomorphism $h$ there exists some decomposition of $\infA$ of some type $t=(s,e)$. A priori there may be two decompositions of one $\omega$-word of two distinct types.

There is an extension of finite-colour Ramsey's theorem to the case where colours form a compact metric space. To state it formally we use the following definitions.

\begin{definition}
Assume that $\seqA = \finA_0, \finA_1,\ldots$ is a sequence of finite words. We say that $\seqB = \finB_0,\finB_1,\ldots$ is a \emph {grouping} of $\seqA$ if there exists an increasing sequence of numbers $0=i_0<i_1<\ldots$ such that for every $n\in\N$ we have
\[\finB_n=\finA_{i_n}\finA_{i_n+1}\ldots\finA_{i_{n+1}-1}.\]
\end{definition}

Note that if $W$ is a decomposition of an $\omega$-word $\infA$ and $W$ is of $h$-type $t=(s,e)$ then every grouping of $W$ is also a decomposition of $\infA$ of $h$-type $t$. The notion of grouping introduces a stronger version of convergence.

\begin{definition}
We say that a sequence of finite words $\seqA=\finA_0,\finA_1,\ldots$ is \emph{strongly convergent} to a profinite word $\finA$ if every grouping of $\seqA$ is convergent to $\finA$.
\end{definition}

The following result can be seen as a simple extension of the Ramsey theorem to the case of the profinite monoid.

\begin{theorem}[Boja{\'n}czyk, Kopczy{\'n}ski, Toru{\'n}czyk~\cite{bojanczyk_ramsey_compact}]\label{th:ramsey_compact}
Let $\seqA=\finA_0,\finA_1,\ldots$ be an infinite sequence of finite words. There exists a grouping $\seqB$ of $\seqA$ such that $\seqB$ strongly converges in $\prom{A}$.
\end{theorem}

For the sake of completeness we give a proof of this fact below. The theorem holds in general, where instead of $\prom{A}$ is any compact metric monoid. Also, the notion of convergence can be strengthened in the thesis of the theorem: all the groupings of $W$ converge in a \emph{uniform way}. In this paper we use only the above, simplified form.

\begin{proof}
Let $K$ be a regular language and $\seqB=\finB_0,\finB_1,\ldots$ be a sequence of finite words. Define a function $\fun{\alpha_{K,\seqB}}{[\N]^2}{\{0,1\}}$ that takes a pair of numbers $i<j$ and returns $1$ if and only if $\finB_i\finB_{i+1}\ldots \finB_{j-1}$ belongs to $K$. By Theorem~\ref{th:ramsey}, there exists a monochromatic set $S\subseteq \N$ with colour $c\in\{0,1\}$ such that for every pair $i<j\in S$ we have $\alpha_{K,\seqB}(\{i,j\})=c$.

Now, take a sequence of finite words $\seqA$. Let $K_0,K_1,\ldots$ be an enumeration of all regular languages and let $\seqA^0=\seqA$. We proceed by induction for $i=0,1,\ldots$. Assume that $\seqA^i=\finA^i_0,\finA^i_1,\ldots$ is defined. First define $\finB_i$ as $\finA^i_0$. Now, let $S=\{n_0,n_1,\ldots\}$ be an infinite monochromatic set with respect to $\alpha_{K_i,\seqA^i}$. Define 
\[\seqA^{i+1}= \left(\finA^i_{n_0}\finA^i_{n_0+1}\ldots\finA^i_{n_1-1}\right),  \left(\finA^i_{n_1}\finA^i_{n_1+1}\ldots\finA^i_{n_2-1}\right),  \left(\finA^i_{n_2}\finA^i_{n_2+1}\ldots\finA^i_{n_3-1}\right), \ldots\]

Note that $\seqA^{i+1}$ is a suffix of a grouping of $\seqA^{i}$. Since $S$ is monochromatic and by the definition of $\alpha_{K,\seqB}$, we know that:\\
$(\ast)$ For every grouping of $\seqA^{i+1}$ either all words in the grouping belong to $K_i$ or all of them do not belong.

We claim that our sequence $\seqB=\finB_0,\finB_1,\ldots$ is strongly convergent. Let $\seqC$ be a grouping of $\seqB$ and let $K=K_i$ be a regular language. Observe that almost all words in $\seqC$ (all except first at most $i$ words) are obtained by grouping words in $\seqA^{i+1}$. Therefore, by $(\ast)$, either almost all words of $\seqC$ belong to $K$ or almost all of them do not belong to $K$. Fact~\ref{ft:conv_seq} implies that $\seqC$ is convergent in $\prom{A}$.

Now observe that almost all words in $\seqC$ belong to $K_i$ if and only if almost all words in $\seqB$ belong to $K_i$. Therefore, the limit of $\seqC$ does not depend on the choice of $\seqC$. It means that $\seqB$ is strongly convergent in $\prom{A}$.
\end{proof}

\subsection{Notation}

In this paper we deal with three types of languages: of finite words, of profinite words, and of $\omega$-words. To simplify reading of the paper, we use the following conventions:
\begin{itemize}
\item finite and profinite words are denoted by $\finA,\finB$,
\item sequences of finite words are denoted by $\seqA,\seqB,\seqC$,
\item $\omega$-words are denoted by $\infA,\infB$,
\item regular languages of finite words are denoted by $K$,
\item profinite-regular languages are, using Theorem~\ref{th:iso}, denoted by $\comp{K}$,
\item general languages of profinite words are denoted by $M$,
\item languages of $\omega$-words (both $\omega$-regular and not) are denoted by $L$.
\end{itemize}

% AUTOMATA =====================================================================

\section{Automata}\label{s:automata}

In this section we provide definitions of four kinds of automata: $\omega$-word models \wB- and \wS-automata and their finite word variants \fB- and \fS-automata. 

The \wB- and \wS-automata models were introduced in~\cite{bojanczyk_bounds}, we follow the definitions from this work. The \fB- and \fS-automata models were defined in~\cite{colcombet_stabilisation}. For the sake of simplicity, we use only the operations $\{\ceps,\cinc,\cres\}$ (without the \emph{check} operation). As noted in Remark~1 in~\cite{colcombet_stabilisation}, this restriction does not influence the expressive power. 

The four automata models we study here are part of a more general theory of regular cost functions that is developed by Colcombet~\cite{colcombet_stabilisation, colcombet_hab}. In particular, the theory of \fB- and \fS-automata has been extended to finite trees in~\cite{colcombet_cost_trees}.

All four automata models we deal with are built on the basis of a \emph{counter automaton}. The difference is the acceptance condition that we introduce later.

\begin{definition}
A \emph{counter automaton} is a tuple $\aut{A}=\left<A, Q, I, \Gamma, \delta\right>$, where:
\begin{itemize}
\item $A$ is an input alphabet,
\item $Q$ is a finite set of states,
\item $I\subseteq Q$ is a set of initial states,
\item $\Gamma$ is a finite set of counters,
\item $\delta\subseteq Q\times A\times \left\{\ceps,\cinc,\cres\right\}^\Gamma\times Q$ is a transition relation.
\end{itemize}
\end{definition}

All counters store natural numbers and cannot be read during a run. The values of the counters are only used in an acceptance condition.

In the initial configuration all counters equal $0$. A transition $(p,a,\cop,q)\in \delta$ (sometimes denoted $p\trans{a,\cop}q$) means that if the automaton is in a state $p$ and reads a letter $a$ then it can perform counter operations $\cop$ and go to the state $q$. For a counter $c\in\Gamma$ a counter operation $\cop(c)$ can:

\vspace{0.2cm}
\begin{tabular}{ l | l }
$\cop(c)=\ceps$ & leave the counter value unchanged, \\
\hline
$\cop(c)=\cinc$ & increment the counter value by one, \\
\hline
$\cop(c)=\cres$ & reset the counter value to $0$.
\end{tabular}
\vspace{0.5cm}

A run $\rho$ of the automaton $\aut{A}$ over a word (finite or infinite) is a sequence of transitions as for standard nondeterministic automata. Given a run $\rho$, a counter $c\in\Gamma$, and a position $r_c$ of a word where the counter $c$ is reset, we define $\val{c, \rho, r_c}$ as the value stored in the counter $c$ at the moment before the reset $r_c$ in $\rho$.

To simplify the constructions we allow $\epsilon$-transitions in our automata. The only requirement is that there is no cycle consisting of $\epsilon$-transitions only. $\epsilon$-transitions can be removed using nondeterminism of an automaton and by combining a sequence of counter operations into one operation. Such a modification may change the exact values of counters, for instance when we replace $\cinc,\cres$ by $\cres$. However, the limitary properties of the counters are preserved (the values may be disturbed only by a linear factor).

\subsection{\wB- and \wS-automata}

First we deal with automata for $\omega$-words, following the definitions in~\cite{bojanczyk_bounds}. An \emph{\wT-automaton} (for $\wT\in\{\wB,\wS\}$) is just a counter automaton. A run $\rho$ of an \wT-automaton over an $\omega$-word $\infA$ is accepting if it starts in an initial state in $I$, every counter is reset infinitely many times, and the following condition is satisfied:
\begin{description}
\item[\wB-automaton] the values of all counters are bounded during the run,
\item[\wS-automaton] for every counter $c$ the values of $c$ during subsequent resets in $\rho$ tend to infinity (i.e. the limit of the values of $c$ is $\infty$).
\end{description}

\noindent An \wT-automaton $\aut{A}$ accepts an $\omega$-word if it has an accepting run on it. The set of all $\omega$-words accepted by $\aut{A}$ is denoted $\lang(\aut{A})$.

\begin{example}
\begin{figure}
\begin{tikzpicture}[shorten >=1pt,node distance=1.5cm,on grid,auto]
\node[state, initial] (qI) at (-3,0) {$q_I$};
\node[state] (qM) at (0,0) {$q_M$};

\path[->] 
    (qI) edge [loop above] node {$a, \ceps$} (qI)
    (qI) edge [loop below] node {$b, \ceps$} (qI)
    (qI) edge [bend left] node {$b, \ceps$} (qM)
    (qM) edge [loop above] node {$a, \cinc$} (qM)
    (qM) edge [bend left] node {$b, \cres$} (qI);
\end{tikzpicture}
\caption{An example of an \wB-automaton $\aut{A}_{\wB}$.}
\label{fig:wB}
\end{figure}

Consider the \wB-automaton $\aut{A}_{\wB}$ depicted on Figure~\ref{fig:wB}. $\aut{A}_{\wB}$ guesses (by moving to the state $q_M$) to measure the length of some blocks of letters $a$. It accepts an $\omega$-word $\infA$ if and only if it is of the form
\[\infA=a^{n_0}b a^{n_1} b\ldots\quad\text{with}\quad\liminf_{i\to\infty} n_i <\infty.\]

We can also treat $\aut{A}_\wB$ as an \wS-automaton. In that case the language recognised by $\aut{A}_{\wB}$ is
\[\{\infA\in\{a,b\}^\omega:\ \infA=a^{n_0}b a^{n_1} b\ldots\quad\text{and}\quad\limsup_{i\to\infty} n_i =\infty\}.\]
\end{example}

It is easy to check that a nondeterministic B\"uchi automaton can be transformed into an equivalent \wB- (resp. \wS)-automaton. Therefore, all $\omega$-regular languages are both \wB- and \wS-regular.
% In this paper we show that the converse also holds (see Corollary~\ref{cl:delta}): a language that is both \wB- and \wS-regular is in fact $\omega$-regular.

\subsection{\fB- and \fS-automata} In the finite word models the situation is a little more complicated than in the \wB- and \wS-automata models. The automaton not only accepts or rejects a given word but also it assigns a \emph{value} to a word.

Formally, a \emph{\fT-automaton} (for $\fT\in\{\fB,\fS\}$) is a counter automaton that is additionally equipped with a set of final states $F\subseteq Q$. An accepting run $\rho$ of an automaton over a finite word $\finA$ is a sequence of transitions starting in some initial state in $I$ and ending in some final state in $F$.

The following equations define $\val{\aut{A}, \finA}$ --- the value assigned to a given finite word by a given automaton. We use the convention that if a set of values is empty then the minimum of this set is $\infty$ and the maximum is $0$. The variable $\rho$ ranges over all accepting runs, $c$ ranges over counters in $\Gamma$, while $r_c$ ranges over positions where the counter $c$ is reset in $\rho$. As noted at the beginning of this section, we do not allow explicit \emph{check} operation, we only care about the values of the counters before resets.

\begin{description}
\item[\fB-automaton $\aut{A}_\fB$] \[\val{\aut{A}_\fB, \finA} = \min_{\rho}\val{\rho}\quad\text{and}\quad\val{\rho}=\max_{c}\max_{r_c}\ \val{c, \rho, r_c},\]
\item[\fS-automaton $\aut{A}_\fS$] \[\val{\aut{A}_\fS, \finA} = \max_{\rho}\val{\rho}\quad\text{and}\quad\val{\rho}=\min_{c}\min_{r_c}\ \val{c, \rho, r_c}.\]
\end{description}

The following simple observation is crucial in the subsequent definitions.

\begin{lemma}\label{lm:val_regular}
For a given number $n$, a \fB-automaton $\aut{A}_\fB$, and an \fS-automaton $\aut{A}_\fS$ the following languages of finite words are regular:
\begin{eqnarray*}
\lang(\aut{A}_\fB\leq n)&=&\left\{w:\val{\aut{A}_\fB, w}\leq n \right\},\\
\lang(\aut{A}_\fS> n)&=&\left\{w:\val{\aut{A}_\fS, w} > n \right\}.
\end{eqnarray*}
\end{lemma}

\begin{proof}
We can encode a bounded valuation of the counters into a state of a finite automaton.
\end{proof}

\subsection{Languages}

The above definitions give a semantics of a \fT-automaton in terms of a function $\fun{\val{\aut{A}, .}}{A^\ast}{\N\cup\{\infty\}}$. As noted in~\cite{torunczyk_limitedness}, it is possible to define the language recognised by such an automaton as a subset of the profinite monoid $\prom{A}$. We successively define it for \fB-automata and \fS-automata. In both cases the construction is justified by Lemma~\ref{lm:val_regular}.

{\bf \fB case:} Fix a \fB-automaton $\aut{A}_\fB$ and define
\begin{equation}
\lang(\aut{A}_\fB):= \bigcup_{n\in\N}\ \comp{\lang(\aut{A}_\fB\leq n})\subseteq \prom{A}.
\label{eq:def_B}
\end{equation}

{\bf \fS case:} Fix an \fS-automaton $\aut{A}_\fS$ and define
\begin{equation}
\lang(\aut{A}_\fS):=\bigcap_{n\in\N}\ \comp{\lang(\aut{A}_\fS>n)}\subseteq \prom{A}.
\label{eq:def_S}
\end{equation}

\noindent Note that the sequences of languages in the above equations are monotone: increasing in \eqref{eq:def_B} and decreasing in \eqref{eq:def_S}.

\newcommand{\vfun}{\mathrm{val}_{\aut{A}}}

There exists another, equivalent way of defining languages recognised by these automata \cite{torunczyk_limitedness}. One can observe that the function $\val{\aut{A}, .}$ assigning to every finite word its value has a unique continuous extension on $\prom{A}$. % Let us denote this extension as $\fun{\vfun}{\prom{A}}{\N\cup\{\infty\}}$. 
The languages recognised by \fB- and \fS-automata can be defined as $\val{\aut{A}, .}^{-1}(\N)$ and $\val{\aut{A}, .}^{-1}(\{\infty\})$ respectively. In this work we only refer to the definitions \eqref{eq:def_B} and \eqref{eq:def_S}.

\begin{example}
Consider the \fS-automaton $\aut{A}_\fS$ depicted in Figure~\ref{fig:S}. The automaton measures the number of letters $a$ in a given word. Then it guesses that the word is finished and moves to the accepting state. For every finite word $\finA$ the value $\val{\aut{A}_\fS,\finA}$ equals the number of letters $a$ in $\finA$.

The language $\lang(\aut{A}_\fS)$ does not contain any finite word. It contains a profinite word $\finA$ if for every $n$ the word $\finA$ belongs to the profinite-regular language defined by the formula ``the word contains more than $n$ letters $a$'' (i.e. $\finA\in\comp{\lang(\aut{A}_\fS > n)}$). In particular, the limit of the sequence $(a^{n!})_{n\in\N}$ from Example~\ref{ex:factorial} belongs to $\lang(\aut{A}_\fS)$.

\begin{figure}
\begin{tikzpicture}[shorten >=1pt,node distance=1.5cm,on grid,auto]
\node[state, initial] (qI) at (-3,0) {$q_I$};
\node[state, accepting] (qF) at (0,0) {$q_F$};

\path[->] 
    (qI) edge [loop above] node {$a, \cinc$} (qI)
    (qI) edge [loop below] node {$b, \ceps$} (qI)
    (qI) edge [] node {$\epsilon, \cres$} (qF);
\end{tikzpicture}
\caption{An example of an \fS-automaton $\aut{A}_\fS$.}
\label{fig:S}
\end{figure}
\end{example}

\begin{lemma}\label{lm:open_closed}
Every \fB-regular language is an open subset of \prom{A} and dually every \fS-regular language is closed.
\end{lemma}

\begin{proof}
By equations~\eqref{eq:def_B} and~\eqref{eq:def_S}, a \fB-regular language is a sum of profinite-regular languages and an \fS-regular language is an intersection of profinite-regular languages. By Fact~\ref{ft:clopen}, profinite-regular languages are closed and open, therefore their sum is open and the intersection is closed.
\end{proof}

The converse of Lemma~\ref{lm:open_closed} is false as there are uncountably many open subsets of \prom{A}.

We finish the definitions of automata models by recalling the following theorem.

\begin{theorem}[Fact~2.6 and Corollary~3.4 in~\cite{bojanczyk_bounds}, Theorem~8 and paragraph \emph{Closure properties} in~\cite{torunczyk_limitedness}]\label{th:effective_disjoint}

Let $\fT\in\{\fB,\fS,\wB,\wS\}$. The class of \fT-regular languages is effectively closed under union and intersection. The emptiness problem for \fT-regular languages is decidable.

Therefore, it is decidable whether given two \fT-regular languages are disjoint.
\end{theorem}

% PROFINITE SEPARATION =======================================================

\section{Separation for profinite languages}\label{s:profinite_sep}

In this section we show the following theorem.

\begin{theorem}\label{th:separation}
Let $\fT\in\{\fB,\fS\}$. Assume that the languages of profinite words $M_1,M_2\subseteq \prom{A}$ are recognised by \fT-automata and $M_1\cap M_2=\emptyset$. Then there exists a profinite-regular language $\comp{K_\sep}\subseteq \prom{A}$ such that
\[M_1\subseteq \comp{K_\sep}\quad\mathrm{and}\quad M_2\subseteq \comp{K_\sep}^c.\]
\end{theorem}

The proof of the theorem consists of two parts, for the two cases of $\fT\in\{\fB,\fS\}$: Lemma~\ref{lm:s_separation} and Theorem~\ref{th:b_separation}.

First we prove the case when $\fT=\fS$. The presented proof uses a general topological fact: the separation property of closed (i.e. $\bpi 1$) sets in a zero-dimensional Polish space.

\begin{lemma}\label{lm:s_separation}
A pair of disjoint \fS-regular languages of profinite words can be separated by a profinite-regular language.
\end{lemma}

\begin{proof}
Take two \fS-regular languages $M_1,M_2\subseteq \prom{A}$.

Since $\prom{A}$ is a zero-dimensional Polish space, the $\bpi 1$-separation property holds for $\prom{A}$ (see Theorem~22.16 in~\cite{kechris_descriptive}). By Lemma~\ref{lm:open_closed} every \fS-regular language is $\bpi 1$ in $\prom{A}$, therefore $M_1,M_2$ can be separated in $\prom{A}$ by a set $M_\sep$ that is closed and open. By Fact~\ref{ft:clopen}, the language $M_\sep$ is profinite-regular.
\end{proof}

Instead of using the $\bpi 1$-separation property, one can provide the following straightforward argument that uses the compactness of $\prom{A}$. We know that $M_1$ is a closed subset of a compact space $\prom{A}$ so $M_1$ is itself compact. Assume that $M_2$ is recognised by an \fS-automaton $\aut{A}_\fS$. By \eqref{eq:def_S} we obtain
\[M_2=\bigcap_{n\in\N}\ \comp{\lang(\aut{A}_\fS>n)}\subseteq \prom{A}.\]

For $n\in\N$ define $N_n:=\comp{\lang(\aut{A}_\fS>n)}^c$ ---  the complement of the profinite-regular language $\comp{\lang(\aut{A}_\fS>n)}$. Clearly $M_1\subseteq \bigcup_n N_n$ because $M_1$ and $M_2$ are disjoint. Fact~\ref{ft:clopen} and Lemma~\ref{lm:val_regular} imply that the sets $N_n$ are open subsets of $\prom{A}$. Therefore, the family $\left(N_n\right)_{n\in\N}$ is an open cover of $M_1$. Since $M_1$ is compact, there is $n_0\in\N$ such that
\[M_1\subseteq N_0\cup N_1\cup\ldots\cup N_{n_0} = N_{n_0}.\]

Therefore, $N_{n_0}$ is a profinite-regular language that separates $M_1$ and $M_2$.

\begin{remark}\label{rm:effect_S}
The language $N_{n_0}$ can be computed effectively.
\end{remark}

\begin{proof}
It is enough to observe that $n_0$ can be taken as the minimal $n$ such that $M_1$ does not intersect the profinite-regular language $\comp{\lang(\aut{A}_\fS>n)}$. Such $n$ exists by the above argument.
\end{proof}

Now we proceed with the separation property for \fB-regular languages. By Lemma~\ref{lm:open_closed} we know that \fB-regular languages are open sets in $\prom{A}$. An easy exercise shows that in general open sets do not have the separation property. Thus, to show the following theorem we need an argument that is a bit more involved than in the case of \fS-regular languages.

\begin{theorem}\label{th:b_separation}
A pair of disjoint \fB-regular languages of profinite words can be separated by a profinite-regular language.
\end{theorem}

We obtain the above theorem by applying the following observation.

\begin{lemma}\label{lm:closed}
For every \fB-regular language $M_\fB\subseteq{\prom{A}}$ there exists a profinite-regular language $\comp{K_R}\subseteq \prom{A}$ such that
\[M_\fB\subseteq \comp{K_R}\quad \mathrm{and}\quad M_\fB\cap A^\ast \ = \ \comp{K_R}\cap A^\ast.\]

Moreover, the language $\comp{K_R}$ can be computed effectively.
\end{lemma}

\begin{proof}
Take a \fB-automaton $\aut{A}_B$ recognising $M_\fB$. Define a new automaton $\aut{A}_R$ by removing from $\aut{A}_B$ all the counters and all the counter operations. What remains are transitions, initial states, and final states. Put $\comp{K_R}=\comp{\lang(\aut{A}_R)}\subseteq\prom{A}$. Of course $M_\fB\subseteq \comp{\lang(\aut{A}_R)}$ by the definition of $M_\fB$. Clearly $\comp{\lang(\aut{A}_R)}\cap A^\ast=\lang(\aut{A}_R)$ by Theorem~\ref{th:iso}. What remains to show is that $\lang(\aut{A}_R)\subseteq M_\fB$. 

Take a finite word $\finA\in \lang(\aut{A}_R)$. Observe that $\aut{A}_B$ has an accepting run on $\finA$ because $\finA\in \lang(\aut{A}_R)$. So $\val{\aut{A}_B, \finA} \leq |\finA|$ because $\aut{A}_B$ cannot do more increments than the number of positions of the word. Therefore $\finA\in M_\fB$.
\end{proof}

\begin{proof}[Proof of Theorem~\ref{th:b_separation}]
Take two disjoint \fB-regular languages $M_1,M_2\subseteq\prom{A}$. Define $\comp{K_\sep}$ to be the language $\comp{K_R}$ from Lemma~\ref{lm:closed} for $M_1$. Thus we know that $M_1\subseteq \comp{K_R}$. We only need to show that $M_2\cap \comp{K_R}=\emptyset$. Assume the contrary, that $M_I := M_2\cap \comp{K_R}\neq \emptyset$. Since \fB-regular languages are open sets in $\prom{A}$, $M_I$ is an open set. Since $A^\ast$ is dense in $\prom{A}$ so $M_I$ contains a finite word $\finA\in A^\ast$. But by the definition of $\comp{K_R}$ in that case $\finA\in M_1$. So $\finA\in M_1\cap M_2$ --- a contradiction to the disjointness of $M_1,M_2$.
\end{proof}

\begin{remark}\label{rm:effect_prof}
Both separation results for \fB- and \fS-regular languages are effective: there is an algorithm that inputs two counter automata, verifies that the intersection of the languages is empty, and outputs an automaton recognising a separating language.
\end{remark}

\begin{proof}
By Theorem~\ref{th:effective_disjoint} it is decidable if two \fB- (resp. \fS)-regular languages are disjoint. As observed in Remark~\ref{rm:effect_S} and Lemma~\ref{lm:closed}, both constructions can be performed effectively.
\end{proof}

% RECOGNITION ==================================================================

\section{Reduction}\label{s:reduction}

This section contains a proof of our crucial technical tool --- Theorem~\ref{th:recognition}. It is inspired by the \emph{reduction theorem} from~\cite{torunczyk_limitedness}.

Intuitively, \wB- and \wS-automata are composed of two \emph{orthogonal} parts, we can call them the $\omega$-regular part and the asymptotic part. The $\omega$-regular part corresponds to states and transitions of the automaton, while the asymptotic part represents quantitative conditions that can be measured by counters. In this section we show how to formally state this division. It can be seen as an extension of the technique presented in~\cite{bojanczyk_bounds}.

\begin{theorem}\label{th:recognition}
Fix an $\wT$-automaton $\aut{A}$ and a type $t=(s,e)$ in the trace monoid $\mtraces(\aut{A})$. There exists a \fT-regular language of profinite words $M_t\subseteq\prom{A}$ with the following property:

If $\infA$ is an $\omega$-word and $\seqA=\finA_0,\finA_1,\ldots$ is a decomposition of $\infA$ of type $t$ then the following conditions are equivalent:
\begin{enumerate}
\item $\infA \in \lang(\aut{A})$,\label{it:in_lang}
\item there exists a grouping $\seqB$ of $\seqA$ that strongly converges to a profinite word $\finB\in M_t$,\label{it:strongly_conv}
\item there exists a grouping $\seqB$ of $\seqA$ that converges to a profinite word $\finB\in M_t$.\label{it:conv}
\end{enumerate}

Additionally, one can ensure that $M_t\subseteq\comp{h_{\aut{A}}^{-1}(e)}$. The construction of a \fT-automaton recognising $M_t$ is effective given $\aut{A}$ and $t$.
\end{theorem}

The rest of this section is devoted to showing the above theorem. We fix for the whole proof an \wT-automaton $\aut{A}=\left<A, Q, I, \Gamma, \delta\right>$ and a type $t=(s,e)$ in $\mtraces(\aut{A})$.

Intuitively, the requirement for a decomposition $\seqA$ to be of the type $t$ corresponds to the $\omega$-regular part of $\aut{A}$ while the convergence of $\seqA$ to an element of $M_t$ takes care of the asymptotic part of $\aut{A}$.

Let us put $K_e=h_{\aut{A}}^{-1}(e)$ and assume that $\aut{B}_e=\left<A, Q_e, \{q_{I,e}\}, \delta_e,F_e\right>$ is a deterministic finite automaton recognising the regular language $K_e$. We will ensure that $M_t\subseteq \comp{K_e}$.

First we show how to construct a language $M_t$, later we prove its properties. The definition of $M_t$ depends on whether $\fT=\fB$ or $\fT=\fS$. The first case is a bit simpler.

{\bf Case $\fT=\fB$} The language $M_t$ is obtained as the union of finitely many \fB-regular languages indexed by states $q\in Q$:
\[M_t=\bigcup_{q\in Q}\ \lang(\aut{A}_q),\]
for \fB-automata $\aut{A}_q$ that we describe below. Intuitively, an automaton $\aut{A}_q$ measures \emph{loops} in $\aut{A}$ starting and ending in $q$.

If for no $q_0\in I$ we have $(q_0, q)\in s$ or if $(q,q)\notin e$ then $\lang(\aut{A}_{q})=\emptyset$. Assume otherwise. First we give an informal definition of $\aut{A}_q$:
\begin{itemize}
\item it is obtained from $\aut{A}$ by interpreting it as a finite word \fB-automaton,
\item it has initial and final state set to $q$,
\item it checks that all the counters are reset in a given word,
\item it checks that a given word belongs to $K_e$,
\item it resets all the counters at the end of the word.
\end{itemize}

Now we give a precise definition of $\aut{A}_q=\left<A, Q_q, I_q, \Gamma_q, \delta_q, F_q\right>$. 
% For the sake of simplicity we allow $\epsilon$ transitions in $\aut{A}_q$.
% To ease the notation, we denote transitions by $q\trans{a,\gamma}q'$ where $q,q'$ are states, $a$ is a letter or $\epsilon$, and $\gamma$ is a vector of counter operations.
Let:
\begin{itemize}
\item $Q_q=\{\ast\}\ \cup\ Q\times Q_e \times \{\bot,\top\}^\Gamma$,
\item $I_q=\left\{\left(q,q_{I,e},(\bot,\bot,\ldots,\bot)\right)\right\}$,
\item $\Gamma_q=\Gamma$,
\item $F_q=\{\ast\}$,
\end{itemize}
and let $\delta_q$ contain the following transitions:
\begin{itemize}
\item $(p,r,b)\trans{a,\cop}(p',r',b')$ if $p\trans{a,\cop}p'\in\delta$, $r\trans{a}r'\in \delta_e$ and for every $c\in\Gamma$ we have $b'(c)=b(c)\lor \left(\cop(c)=\cres\right)$,
\item $\left(q,r,(\top,\top,\ldots,\top)\right)\trans{\epsilon,\cop}\ast$ for $\cop=(\cres,\cres,\ldots,\cres)$ if $r\in F_e$.
\end{itemize}

The state $\ast$ is the only final state used to perform the reset at the end of a word. During a run, the automaton $\aut{A}_q$ simulates $\aut{A}$ and $\aut{B}_e$ in parallel, using $Q$ and $Q_e$. Additionally, a vector in $\{\bot,\top\}^\Gamma$ denotes for every counter whether it was already reset in a word or not.

{\bf Case $\fT=\fS$} In that case the language $M_{t}$ is obtained as the union of finitely many \fS-regular languages indexed by pairs $(q,\tau)\in Q\times \{\shl,\shr\}^\Gamma$:
\[M_{t}=\bigcup_{(q,\tau)}\ \lang(\aut{A}_{q,\tau}).\]
Intuitively, an automaton $\aut{A}_{q,\tau}$ recognises loops $q\to^\ast q$ as before. Additionally, the vector $\tau$ denotes whether a given counter $c\in\Gamma$ obtains bigger values before the first reset ($\tau(c)=\shr$) or after the last reset ($\tau(c)=\shl$) on a given finite word. The following definition formalises this property. A similar technique of assigning a \emph{reset type} to a finite run can be found in~\cite{bojanczyk_bounds}.

\begin{definition}\label{def:etype}
Let $\rho$ be a run of some counter automaton $\aut{A}$ over an $\omega$-word $\infA$. Let $k\in\N$ be a position in $\infA$ and let $c\in\Gamma$ be a counter of $\aut{A}$. Let:

\begin{itemize}
\item $V_L$ be the number of increments of $c$ between the last reset before $k$ and $k$,
\item $V_R$ be the number of increments of $c$ between $k$ and the first reset after $k$.
\end{itemize}
If there is no reset of $c$ at some side of $k$ then the respective value is $0$. Define the \emph{end-type} of $c$ on $\rho$ in $k$ (denoted as $\etype (c,\rho, k)$) by the following equation:

\[\etype(c,\rho, k)\ =
\begin{cases} \shr & \text{if $V_L< V_R$,}\\
\shl &\text{if $V_L\geq V_R$.}
\end{cases} \]
\end{definition}

As before if for no $q_0\in I$, we have $(q_0, q)\in s$ or if $(q,q)\notin e$ then $\lang(\aut{A}_{q,\tau})=\emptyset$. Assume otherwise. We start with an informal definition of $\aut{A}_{q,\tau}$:
\begin{itemize}
\item it is obtained from $\aut{A}$ by interpreting it as a finite word \fS-automaton,
\item it has initial and final state set to $q$,
\item it checks that all the counters are reset in a given word,
\item it checks that a given word belongs to $K_e$,
\item for every counter $c\in\Gamma$:
\begin{itemize}
\item if $\tau(c)=\shl$ then $A_{q,\tau}$ skips the first reset of $c$ and all the previous increments of $c$ but resets $c$ at the end of a given word,
\item if $\tau(c)=\shr$ then $A_{q,\tau}$ acts on $c$ exactly as $\aut{A}$ (with no additional reset at the end of the word).
\end{itemize}
\end{itemize}

Formally, let $\aut{A}_{q,\tau}=\left<A, Q_{q,\tau}, I_{q,\tau}, \Gamma_{q,\tau}, \delta_{q,\tau}, F_{q,\tau}\right>$ such that
\begin{itemize}
\item $Q_{q,\tau}=\{\ast\}\ \cup\ Q\times Q_e \times \{\bot,\top\}^\Gamma$,
\item $I_{q,\tau}=\left\{\left(q,q_{I,e},(\bot,\bot,\ldots,\bot)\right)\right\}$,
\item $\Gamma_{q,\tau}=\Gamma$,
\item $F_{q,\tau}=\{\ast\}$,
\end{itemize}
and $\delta_{q,\tau}$ contains the following transitions:
\begin{itemize}
\item $(p,r,b)\trans{a,\cop'}(p',r',b')$ if $p\trans{a,\cop}p'\in\delta$, $r\trans{a}r'\in \delta_e$, and for every $c\in\Gamma$ we have:
\begin{itemize}
\item $b'(c)=b(c)\lor \left(\cop(c)=\cres\right)$,
\item if $b(c)=\bot$ and $\tau(c)=\shl$ then $\cop'(c)=\ceps$, otherwise $\cop'(c)=\cop(c)$,
\end{itemize}
\item $\left(q,r,(\top,\top,\ldots,\top)\right)\trans{\epsilon,\cop}\ast$ if $r\in F_e$ and for every $c\in\Gamma$ we have $\cop(c)=\cres$ if $\tau(c)=\shl$ and $\cop(c)=\ceps$ otherwise.
\end{itemize}

\noindent Now we proceed with the proof that the above constructions give us the desired language $M_t$. First note that in both cases the constructed automata explicitly verify that a given word belongs to $K_e$. Therefore, $M_t\subseteq\comp{K_e}$.

We start by taking an $\omega$-word $\infA$ and its decomposition $\seqA=\finA_0,\finA_1,\ldots$ of the type $t$.

\subsection{Implication $(\ref{it:in_lang}) \Rightarrow (\ref{it:strongly_conv})$}

Assume that there exists an accepting run $\rho$ of $\aut{A}$ over $\infA$. We want to construct a grouping $\seqB=\finB_0,\finB_1,\ldots$ of $\seqA$ such that:
\begin{enumerate}[label=S.\arabic*]
\item for $n>0$ we have $\finB_n\in K_e$,\label{en:type}
\item all counters in $\Gamma$ are reset by $\rho$ in every word $\finB_n$,\label{en:counters}
\item the state that occurs in the run $\rho$ at the end-points of all the words $\finB_n$ is some fixed state $q\in Q$,\label{en:state}
\item there exists a vector $\tau\in \{\shl,\shr\}^\Gamma$ such that for every counter $c$ and every position $k$ between successive words $\finB_n,\finB_{n+1}$ in $\infA$ we have $\etype (c, \rho, k)=\tau(c)$,\label{en:vector}
\item the sequence of words $\seqB$ is strongly convergent to some profinite word $\finB$.\label{en:conv}
\end{enumerate}

The grouping $Z$ is obtained in steps. Observe that all the above properties are preserved when taking a grouping of a sequence. \ref{en:type} is already satisfied by the sequence $\seqA$. First, we group words of $W$ in such a way to satisfy \ref{en:counters} using the fact that the run $\rho$ is accepting. Then we further group the sequence to satisfy \ref{en:state} and \ref{en:vector} --- some state and value of $\etype$ must appear in infinitely many end-points. Finally, we apply Theorem~\ref{th:ramsey_compact} to group the sequence into a strongly convergent one. 

Now, it suffices to show that $\finB\in M_t$. First, observe that $\rho$ is a witness that there is a path from $I$ to $q$ and from $q$ to $q$ in $\aut{A}$.

We consider two cases:
\begin{description}
\item[\bf Case $\fT=\fB$] Since $\rho$ is accepting, there exists a constant $l$ such that the values of all counters during $\rho$ are bounded by $l$. We show that for every $n>0$ we have $\finB_n\in\lang(\aut{A}_q\leq l)$. It implies that $\finB\in \comp{\lang(\aut{A}_q\leq l)}$ and therefore $\finB\in\lang(\aut{A}_q)\subseteq M_t$.

Observe that $\rho$ induces a run $\rho_n$ of $\aut{A}_q$ on $\finB_n$. By \ref{en:type}, \ref{en:counters}, and \ref{en:state} we know that $\rho_n$ is an accepting run of $\aut{A}_q$ --- it starts in the only initial state and ends in $\ast$. Since $\aut{A}_q$ simulates all the resets of $\aut{A}$, we know that $\val{\rho_n}\leq l$ and therefore $\val{\aut{A}_q,\finB_n}\leq l$.

\item[\bf Case $\fT=\fS$] We show that for every $l\in\N$ the sequence $\seqB$ from some point on satisfies $\val{\aut{A}_{q,\tau}, \finB_n}> \frac{l}{2}$. It implies that for every $l$ we have $\finB\in\comp{\lang(\aut{A}_{q,\tau}> l)}$ and therefore $\finB\in\lang(\aut{A}_{q,\tau})$.

Since $\rho$ is accepting, for every constant $l$, from some point on, all the counters are reset with a value greater than $l$. Assume that the last reset with the value at most $l$ occurs before the word $\finB_N$. We show that for $n\geq N$ we have $\val{\aut{A}_{q,\tau}, \finB_n}> \frac{l}{2}$. Let $\rho_n'$ be the sequence of transitions of $\rho$ on $\finB_n$. Observe that $\rho_n'$ induces a run $\rho_n$ of $\aut{A}_{q,\tau}$ on $\finB_n$. As before, $\rho_n$ is accepting by \ref{en:type}, \ref{en:state}, and \ref{en:counters}. Take a counter $c\in\Gamma$ and a reset of this counter $r_c$ in $\rho_n$. Consider the following cases, recalling Definition~\ref{def:etype}:

\begin{itemize}
\item $r_c$ corresponds to the first reset of $c$ in the run $\rho_n'$. Since $\aut{A}_{q,\tau}$ did not skip $r_c$, $\tau(c)=\shr$. Therefore, $c$ has more increments after the beginning of $\finB_n$ than before it in $\rho$. Therefore $\val{c,\rho_n,r_c}>\frac{l}{2}$.
\item $r_c$ corresponds to a reset of $c$ in the run $\rho_n'$ but not the first one. In that case $\val{c,\rho_n,r_c}=\val{c,\rho_n',r_c}>l$.
\item $r_c$ is the additional reset performed by $\aut{A}_{q,\tau}$ at the end of the word $\finB_n$. In that case $\tau(c)=\shl$ so $c$ has greater or equal number of increments before the end of the word $\finB_n$ than after it in $\rho$. Therefore $\val{c,\rho_n,r_c}>\frac{l}{2}$.
\end{itemize}

In all three cases $\val{c,\rho_n,r_c}> \frac{l}{2}$. So we have shown that
\[\val{\aut{A}_{q,\tau}}\geq\val{\rho_n}>\frac{l}{2}.\]
\end{description}

\subsection{Implication $(\ref{it:strongly_conv}) \Rightarrow (\ref{it:conv})$}

This implication is trivial since strong convergence entails convergence.

\subsection{Implication $(\ref{it:conv}) \Rightarrow (\ref{it:in_lang})$}

Let $\seqB$ be a grouping of $\seqA$ such that $\seqB$ converges to a limit $\finB\in M_t$.

We consider two cases:
\begin{description}
\item[\bf Case $\fT=\fB$] Since $\finB\in M_t$, there exists a state $q\in Q$ such that $\finB\in \lang(\aut{A}_{q})$. Therefore, $\finB\in\comp{\lang(\aut{A}_q\leq l)}$ for some $l$. Since $\comp{\lang(\aut{A}_q\leq l)}$ is an open set and $\finB$ is a limit of $\seqB$, almost all elements of $\seqB$ belong to $\lang(\aut{A}_q\leq l)$. Assume that for $n\geq N$ we have $\finB_n\in\lang(\aut{A}_q\leq l)$. Let $\rho_n$ be a run that witnesses this fact. By the construction of $\aut{A}_q$, the run $\rho_n$ induces a run $\rho_n'$ of $\aut{A}$ on $\finB_n$. Also, since $\rho_n$ is accepting, $\rho_n'$ resets all the counters at least once.

By the assumption about $t$, there exists a run $\rho_0'$ of $\aut{A}$ on $\finB_0$ that starts in some state in $I$ and ends in $q$, and a sequence of runs $\rho_n'$ on $\finB_n$ for $0<n< N$ that lead from $q$ to $q$. Therefore, we can construct an infinite run $\rho$ of $\aut{A}$ on $\infA$ being the concatenation of the runs $\rho_n'$ on the words $\finB_n$ for $n\in\N$. We show that if $r_c$ is a reset of a counter $c$ in $\rho$ that appears after the word $\finB_N$ then $\val{c,\rho,r_c}\leq 2\cdot l$. Since there are only finitely many resets of counters before the word $\finB_N$, this bound suffices to show that the run $\rho$ is accepting.

Observe that the increments in $\rho$ correspond to the increments in the runs $\rho_n$. Also, $\rho$ performs all the resets that appear in runs $\rho_n$ except the resets at the end of the words. There can be at most one such skipped reset in a row because every counter is reset in every run $\rho_n'$. Therefore, $\val{c,\rho,r_c}\leq 2\cdot l$.

%Therefore, $\rho$ is an accepting run and $\infA\in \lang(\aut{A})$.

\item[\bf Case $\fT=\fS$] Let $q,\tau$ be parameters such that $\finB\in \lang(\aut{A}_{q,\tau})$. Therefore, for every $l\in\N$ we have $\finB\in\comp{\lang(\aut{A}_{q,\tau}> l)}$. As $\seqB$ is convergent to $\finB$ and languages $\comp{\lang(\aut{A}_{q,\tau}> l)}$ are open, it means that
\begin{equation}
\forall_l\ \exists_N\ \forall_{n\geq N}\ \val{\aut{A}_{q,\tau}, \finB_n}> l.
\label{eq:limits}
\end{equation}

As above we construct a run $\rho$ over $\infA$ that first leads on $\finB_0$ from some state of $I$ to $q$ and later consists of a concatenation of runs over words $\finB_n$. Let $\rho_0'$ be any run of $\aut{A}$ that leads from $I$ to $q$ on $\finB_0$. For $n>0$ we pick a run $\rho_n$ in such a way that it is accepting and\footnote{Since there are only finitely many runs of an automaton on a finite word, there always exists a run realising the value $\val{\aut{A}_{q,\tau},\finB_n}$, no matter whether the value is finite or not.}
\[\val{\rho_n}= \val{\aut{A}_{q,\tau},\finB_n}.\]
Observe that by \eqref{eq:limits}, we obtain
\begin{equation}
\lim_{n\to\infty}\val{\aut{A}_{q,\tau},\finB_n} =\lim_{n\to\infty}\val{\rho_n}=\infty.
\label{eq:tends}
\end{equation}

For $n>0$ by $\rho_n'$ be denote the run of $\aut{A}$ on $\finB_n$ induced by $\rho_n$. Similarly as in the previous case, runs $\rho_n'$ for $n\in\N$ can be combined into a run $\rho$ of $\aut{A}$ on $\infA$. By the construction of $\aut{A}_{q,\tau}$, $\rho$ resets every counter infinitely often.

Let $r_c$ be a position in $\infA$ where a counter $c\in\Gamma$ is reset during $\rho$. Assume that $r_c$ is contained in a word $\finB_n$ and $n>1$ --- we do not care about first two words.

Consider two cases:
\begin{description}
\item[$(\tau(c)=\shr)$] In that case $\rho$ performs the same increments and resets of $c$ as the runs $\rho_n$. Therefore, $\val{c,\rho,r_c}\geq \val{\rho_n}$.
\item[$(\tau(c)=\shl)$] If $r_c$ is not the first reset of $c$ in $\rho_n'$ then the value of $c$ before $r_c$ in $\rho$ is the same as in $\rho_n$. Assume that $r_c$ is the first reset of $c$ in $\rho_n'$. Note that $\rho_{n-1}$ performs an additional reset of $c$ at the end of $\finB_{n-1}$. This reset does not appear in $\rho$ so $\val{c,\rho,r_c}\geq \val{\rho_{n-1}}$.
\end{description}

In all the cases
\[\val{c,\rho,r_c}\geq \min\left(\val{\rho_{n-1}},\val{\rho_n}\right),\]
so the values of $c$ before successive resets tend to infinity by \eqref{eq:tends}. It means that $\rho$ is an accepting run and $\infA\in\lang(\aut{A})$.
\end{description}

% SEPARATION ===================================================================

\section{Separation for $\omega$-languages}\label{s:omega_sep}

In this section we show the main result of the paper. The technique is to lift the separation results for \fT-regular languages of profinite words into the $\omega$-word case.

\begin{theorem}\label{th:main}
Let $\fT\in\{\fB,\fS\}$. If $L_1,L_2$ are disjoint languages of $\omega$-words both recognised by \wT-automata then there exists an $\omega$-regular language $L_\sep$ such that
\[L_1\subseteq L_\sep\quad\mathrm{and}\quad L_2\subseteq L_\sep^c.\]

Additionally, the construction is effective.
\end{theorem}

The rest of the section is devoted to showing this theorem. As observed by Thomas Colcombet, in the case of $\fT=\fB$ the thesis can be proved directly, without referring to profinite words. This simpler proof is presented in Section~\ref{s:direct} with his kind permission. However, the $\fT=\fB$ case is also treated here for two reasons: first it reveals the symmetry and generality of the ``profinite approach'', second it can be used as a guideline for the more complex case of $\fT=\fS$.

Let $i\in\{1,2\}$ and $\mtraces^i$ denote the trace-monoid for an \wT-automaton $\aut{A}_i$ recognising $L_i$. Let $h_i=h^i_{\aut{A}_i}$ be the canonical homomorphisms from $\prom{A}$ to $\mtraces^i$. Define $\types{i}$ as the set of types $t_i=(s_i,e_i)$ in the trace-monoid $\mtraces^i$.

For every type $t_i=(s_i,e_i)\in \types{i}$ define $M_{t_i}^i\subseteq\prom{A}$ as the \fT-regular language of profinite words given by Theorem~\ref{th:recognition} for $\aut{A}=\aut{A}_i$ and $t=t_i$. By the statement of the theorem we know that $M_{t_i}^i\subseteq \comp{h_i^{-1}(e_i)}$.

\begin{definition}
For a pair of types $t_1=(s_1,e_1)\in \types{1}, t_2=(s_2,e_2)\in \types{2}$, we say that $t_1,t_2$ are \emph{coherent} if there exist finite words $\finA_s, \finA_e\in A^\ast$ such that: $h_i(\finA_s)=s_i$ and $h_i(\finA_e)=e_i$ for $i=1,2$.
\end{definition}

An important application of Theorem~\ref{th:recognition} is the following lemma.

\begin{lemma}
If a pair of types $t_1\in\types{1},t_2\in\types{2}$ is coherent then the languages $M_{t_1}^1,M_{t_2}^2$ are disjoint.
\end{lemma}

\begin{proof}
Take coherent types $t_1=(s_1,e_1)$ and $t_2=(s_2,e_2)$.

Assume that there exists a profinite word $\finA\in M_{t_1}^1\cap M_{t_2}^2$. Since $\finA\in\comp{h_i^{-1}(e_i)}$ for $i=1,2$, there exists a sequence $\seqA=\finA_1,\finA_2,\ldots$ of finite words converging to $\finA$ such that $h_1(\finA_n)=e_1$ and $h_2(\finA_n)=e_2$ for all $n>0$. Moreover, by coherency of $t_1,t_2$ there exists a finite word $\finA_0$ such that $h_1(\finA_0)=s_1$ and $h_2(\finA_0)=s_2$. Let $\infA=\finA_0 \finA_1 \finA_2\ldots$ We show that $\infA\in L_1\cap L_2$ --- a contradiction.

Take $i\in\{1,2\}$. Observe that $\infA=\finA_0 \finA_1\ldots$ is a decomposition of $\infA$ of $h_i$-type $t_i$. Additionally observe that the sequence $\seqA$ converges to $\finA$ and $\finA$ belongs to $M_{t_i}^i$. So, by Theorem~\ref{th:recognition} we have $\infA\in L_i$.
\end{proof}

Take a pair of coherent types $t_1,t_2$. Since the languages $M_{t_1}^1,M_{t_2}^2$ are disjoint, we can use Theorem~\ref{th:separation} to find a separating profinite-regular language $\comp{R_{t_1,t_2}}\subseteq\prom{A}$ such that
\[M_{t_1}^1\subseteq \comp{R_{t_1,t_2}}\quad\text{and}\quad M_{t_2}^2\subseteq \comp{R_{t_1,t_2}}^c.\]

Now we can introduce the $\omega$-regular language $L_\sep$ separating $L_1,L_2$.

\begin{definition}
Consider a coherent pair of types $(t_1,t_2)$. Let $S_{t_1,t_2}$ be defined as follows: $S_{t_1,t_2}$ is the language of $\omega$-words $\infA$ such that there exists a decomposition $\infA=\finA_0\finA_1\ldots$ of types $t_1,t_2$ with respect to $h_1,h_2$, such that every grouping of $(\finA_n)_{n\in\N}$ from some point on belongs to the regular language $R_{t_1,t_2}$.

Note that the above definition can be expressed in MSO so $S_{t_1,t_2}$ is an $\omega$-regular language.

Let $L_\sep$ be the $\omega$-regular language defined as
\[L_\sep=\bigcup_{t_1,t_2}\ S_{t_1,t_2},\]
where the sum ranges over pairs of coherent types.
\end{definition}

Clearly $L_\sep$ is an $\omega$-regular language. What remains is to show the following lemma.

\begin{lemma}
The language $L_\sep$ separates $L_1$ and $L_2$.
\end{lemma}

\begin{proof}
First observe that $L_1\subseteq L_\sep$. Take $\infA\in L_1$. We want to construct a decomposition $\seqA=\finA_0,\finA_1,\ldots$ of $\infA$ such that:
\begin{itemize}
\item the $h^i$-type of $\seqA$ is $t_i$ for $i=1,2$ and some pair of coherent types $(t_1,t_2)$ in $\types{1}\times\types{2}$,
\item the sequence $\seqA$ is strongly convergent to some profinite word $\finA\in\prom{A}$.
\end{itemize}

The sequence $\seqA$ is obtained in steps. First we use Theeorem~\ref{th:ramsey} to find a decomposition of $\infA$ with respect to both monoids $\mtraces^1,\mtraces^2$ at the same time. Such decomposition satisfies the first bullet above. Then, using Theorem~\ref{th:ramsey_compact}, we can group our sequence into $\seqA$ in such a way that $\seqA$ is strongly convergent.

By Theorem~\ref{th:recognition}, there exists a grouping $\seqB$ of $\seqA$ that converges to a profinite word $\finB\in M_{t_1}^1\subseteq \comp{R_{t_1,t_2}}$. But since $\seqA$ is strongly convergent, $\finB=\finA$. Therefore, by the strong convergence of $\seqA$, every grouping of $\seqA$ converges to $\finA\in \comp{R_{t_1,t_2}}$. So every grouping of $\infA$ from some point on belongs to $R_{t_1,t_2}$ as in the definition of $L_\sep$. Therefore, $\infA\in L_\sep$.

Now we show that $L_2\cap L_\sep=\emptyset$. Assume otherwise, that there exists an $\omega$-word $\infA\in L_2\cap L_\sep$. Since $\infA\in L_\sep$, there exists a coherent pair of types $t_1,t_2$ such that $\infA\in S_{t_1,t_2}$. Therefore, $\infA$ can be decomposed as $\infA=\finA_0\finA_1\ldots$ of types $t_1,t_2$ respectively. Let $\seqA=\finA_0,\finA_1,\ldots$ Because $\infA\in L_2$ so by Theorem~\ref{th:recognition} there exists a grouping $\seqB$ of $\seqA$ with a limit $\finB\in M_{t_2}^2$ . But by the definition of $S_{t_1,t_2}$ almost all words in $\seqB$ belong to $R_{t_1,t_2}$ so $\finB\in \comp{R_{t_1,t_2}}$. Since $\comp{R_{t_1,t_2}}\cap M_{t_2}^2=\emptyset$, we have the required contradiction.
\end{proof}

Now we can deduce the corollary from the introduction.

\begin{corollary}\label{cl:delta}
If a given language of $\omega$-words $L$ and its complement $L^c$ are both \wB-regular (resp. \wS-regular) then $L$ is (effectively) $\omega$-regular.
\end{corollary}

\begin{proof}
Let $L$ be a language of $\omega$-words such that $L$ and $L^c$ are both $\wT$-regular. By Theorem~\ref{th:main} there exists an $\omega$-regular language $L_\sep$ that separates $L$ and $L^c$. But in that case $L_\sep=L$ so $L$ is $\omega$-regular.
\end{proof}

% ACKNOWLEDGEMENTS =============================================================

\section{A direct proof of separation for \wB-regular languages}\label{s:direct}

As observed by Thomas Colcombet, the separation property for \wB-regular languages can be shown directly, without referring to profinite words. This simpler proof is presented here with his kind permission.

Let $\lang(\aut{A})$ be an \wB-regular language recognised by an \wB-automaton $\aut{A}$. Consider a B\"uchi automaton $\aut{A}'$ obtained from $\aut{A}$ by removing all the counter operations (similarly to Lemma~\ref{lm:closed}) and requiring that every counter is reset infinitely often.

Clearly, the language recognized by $\aut{A}'$ is $\omega$-regular and $\lang(\aut{A})\subseteq\lang(\aut{A}')$.

\begin{claim}\label{cl:ultimat}
If $\infA=\finA\finB\finB\ldots$ is an ultimately periodic $\omega$-word in $\lang(\aut{A}')$ then $\infA\in \lang(\aut{A})$.
\end{claim}

\begin{proof}
Observe that $\aut{A}'$ has an ultimately periodic accepting run $\rho$ on $\infA$. By the acceptance condition of $\aut{A}'$, every counter $c$ of $\aut{A}$ is reset infinitely often during $\rho$. Since $\rho$ is ultimately periodic, the values of the counter $c$ are bounded in $\rho$. Therefore, $\infA\in \lang(\aut{A})$.
\end{proof}

It means that, since an $\omega$-regular language is entirely defined by the ultimately periodic $\omega$-words it contains~\cite{buchi_decision}, $\lang(\aut{A}')$ is the least $\omega$-regular language that contains $\lang(\aut{A})$. It also means that it depends only on $\lang(\aut{A})$ but not on the specific automaton $\aut{A}$ that recognizes it. Let us call this language $\closure(\lang(\aut{A}))$.

Consider now two \wB-regular languages of empty intersection $L_1$ and $L_2$. Assume $\closure(L_1)$ intersects $\closure(L_2)$ then, since these languages are $\omega$-regular, there is an ultimately periodic $\omega$-word in this intersection. But according to Claim~\ref{cl:ultimat}, this ultimately periodic $\omega$-word belongs to both $L_1$ and $L_2$. A contradiction.

It follows that if $L_1$ and $L_2$ are disjoint then $\closure(L_1)$ (respectively $\closure(L_2)$) are separators. Also, this construction shows that in order to construct a separator of two \wB-regular languages, only one language needs to be known.

% ACKNOWLEDGEMENTS =============================================================

\section{Acknowledgements}\label{s:ack}

The author would like to thank Miko{\l}aj Boja{\'n}czyk for his suggestions. This paper makes noticeable use of results by Szymon Toru{\'n}czyk: his PhD thesis~\cite{torunczyk_phd} and paper~\cite{torunczyk_limitedness}. The author would like to thank Filip Mazowiecki and the referees for careful reading the text and providing a number of important comments. Moreover, Thomas Colcombet had a gainful influence on the paper by providing a direct proof for the separation of \wB-regular languages.

% BIBLIOGRAPHY =================================================================

\bibliographystyle{alpha}
\bibliography{mskrzypczak}

\begin{thebibliography}{HMN09}

\bibitem[Alm03]{almeida_profinite}
Jorge Almeida.
\newblock Profinite semigroups and applications.
\newblock In {\em Structural Theory of Automata, Semigroups, and Universal
  Algebra}, pages 7--18, 2003.

\bibitem[AMN12]{michalewski_separation}
Andr{\'e} Arnold, Henryk Michalewski, and Damian Niwi{\'n}ski.
\newblock On the separation question for tree languages.
\newblock In {\em STACS}, pages 396--407, 2012.

\bibitem[BC06]{bojanczyk_bounds}
Miko{\l}aj Boja{\'n}czyk and Thomas Colcombet.
\newblock Bounds in $\omega$-regularity.
\newblock In {\em LICS}, pages 285--296, 2006.

\bibitem[BKT12]{bojanczyk_ramsey_compact}
Miko{\l}aj Boja{\'n}czyk, Eryk Kopczy{\'n}ski, and Szymon Toru{\'n}czyk.
\newblock {Ramsey}’s theorem for colors from a metric space.
\newblock {\em Semigroup Forum}, 85:182--184, 2012.

\bibitem[B{\"u}c62]{buchi_decision}
Julius~Richard B{\"u}chi.
\newblock On a decision method in restricted second-order arithmetic.
\newblock In {\em Proc. 1960 Int. Congr. for Logic, Methodology and Philosophy
  of Science}, pages 1--11, 1962.

\bibitem[CL10]{colcombet_cost_trees}
Thomas Colcombet and Christof L{\"o}ding.
\newblock Regular cost functions over finite trees.
\newblock In {\em LICS}, pages 70--79, 2010.

\bibitem[Col09]{colcombet_stabilisation}
Thomas Colcombet.
\newblock The theory of stabilisation monoids and regular cost functions.
\newblock In {\em ICALP (2)}, pages 139--150, 2009.

\bibitem[Col13]{colcombet_hab}
Thomas Colcombet.
\newblock Fonctions r{\'e}guli{\`e}res de co{\^u}t.
\newblock Habilitation thesis, Universit{\'e} Paris Diderot---Paris 7, 2013.

\bibitem[HMN09]{hummel_separation}
Szczepan Hummel, Henryk Michalewski, and Damian Niwi{\'n}ski.
\newblock On the {Borel} inseparability of game tree languages.
\newblock In {\em STACS}, pages 565--575, 2009.

\bibitem[Kec95]{kechris_descriptive}
Alexander Kechris.
\newblock {\em Classical descriptive set theory}.
\newblock Springer-Verlag, New York, 1995.

\bibitem[KV99]{kupferman_complementation}
Orna Kupferman and Moshe~Y. Vardi.
\newblock The weakness of self-complementation.
\newblock In {\em STACS}, pages 455--466, 1999.

\bibitem[Pin09]{pin_profinite}
{Jean-\'{E}ric} Pin.
\newblock Profinite methods in automata theory.
\newblock In {\em STACS}, pages 31--50, 2009.

\bibitem[PP04]{perrin_pin_words}
Dominique Perrin and {Jean-\'{E}ric} Pin.
\newblock {\em Infinite Words: Automata, Semigroups, Logic and Games}.
\newblock Elsevier, 2004.

\bibitem[Rab70]{rabin_separation}
Michael~O. Rabin.
\newblock Weakly definable relations and special automata.
\newblock In {\em Proceedings of the Symposium on Mathematical Logic and
  Foundations of Set Theory}, pages 1--23. North-Holland, 1970.

\bibitem[Tor11]{torunczyk_phd}
Szymon Toru{\'n}czyk.
\newblock {\em Languages of profinite words and the limitedness problem}.
\newblock PhD thesis, University of Warsaw, 2011.

\bibitem[Tor12]{torunczyk_limitedness}
Szymon Toru{\'n}czyk.
\newblock Languages of profinite words and the limitedness problem.
\newblock In Artur Czumaj, Kurt Mehlhorn, Andrew~M. Pitts, and Roger
  Wattenhofer, editors, {\em ICALP (2)}, volume 7392 of {\em Lecture Notes in
  Computer Science}, pages 377--389. Springer, 2012.

\end{thebibliography}

\end{document}